%% file: main.tex
\documentclass{article}

\usepackage{nips_2017}

\usepackage{times}
\usepackage{titlesec}
\usepackage{hyperref}
\usepackage{url}
\usepackage{gensymb}
\usepackage{multirow}
\usepackage{amsmath}
\usepackage[export]{adjustbox}
\usepackage{subfig}
\usepackage{caption}
\usepackage{algorithmic}
\usepackage[skins]{tcolorbox}
\usepackage{kantlipsum, tabularx, tikz, booktabs}
\usepackage{varwidth}
\usepackage{relsize}
\usepackage{times}
\usepackage{float}
\usepackage{color}
\usepackage{soul}
\usepackage{setspace}
\usepackage[normalem]{ulem}
\usepackage{color}
\usepackage{multicol}
\usepackage{amsthm}
\usepackage{hyperref}
\usepackage{cleveref}
\usepackage{lipsum}
\usepackage{wrapfig}
\usepackage{gensymb}
\usepackage[utf8]{inputenc} 
\usepackage[T1]{fontenc}    
\usepackage{hyperref}       
\usepackage{url}            
\usepackage{booktabs}       
\usepackage{amsfonts}       
\usepackage{nicefrac}       
\usepackage{microtype}      
\usepackage{wrapfig}
\usepackage{amsmath}
\usepackage{amsthm}
\usepackage{amssymb}
\usepackage{graphicx}
\usepackage[colorinlistoftodos]{todonotes}
\usepackage[ruled, linesnumbered]{algorithm2e}
\usepackage{caption}
\usepackage{float}
\usetikzlibrary{arrows}
\tikzset{
  treenode/.style = {align=center, inner sep=0pt, text centered,
    font=\sffamily},
  arn_n/.style = {treenode, circle, white, font=\sffamily\bfseries, draw=black,
    fill=black, text width=1.5em},
  arn_r/.style = {treenode, circle, red, draw=red, 
    text width=1.5em, very thick},
  arn_x/.style = {treenode, rectangle, draw=black,
    minimum width=0.5em, minimum height=0.5em}
}
\newtheorem{theorem}{Theorem}

\newtheorem*{lemma*}{Lemma}	
\newtheorem*{theorem*}{Theorem}

\theoremstyle{definition}
\newtheorem{defn}{Definition}[]

\titlespacing\section{2pt}{12pt plus 4pt minus 2pt}{2pt plus 2pt minus 2pt}
\titlespacing\subsection{2pt}{8pt plus 4pt minus 2pt}{2pt plus 2pt minus 2pt}
\titlespacing\subsubsection{2pt}{8pt plus 4pt minus 2pt}{2pt plus 2pt minus 2pt}

\SetKwComment{Comment}{$\triangleright$\ }{}
\usepackage[margin=1cm]{caption}

\newcommand{\argminF}{\mathop{\mathrm{argmin}}\limits}   
\newcommand{\squeezeup}{\vspace{-2.5mm}}

\usetikzlibrary{shapes.arrows}
\usetikzlibrary{arrows.meta}
\usetikzlibrary{patterns}

\newcommand{\citeravi}[1]{(\citeauthor{#1}, \citeyear{#1})}
\newcommand{\citenopar}[1]{\citeauthor{#1}, \citeyear{#1}}

\newcommand{\del}[1]{}

\newcommand{\clearrow}{\global\let\rowmac\relax}
\newcommand{\edge}[2]{$e_{(#1,#2)}$}
\newcommand{\edgem}[2]{e_{(#1,#2)}}
\clearrow

\author{
}

\title{PURE: Scalable Phase Unwrapping with Spatial Redundant Arcs}

\begin{document}

\maketitle
\begin{abstract}
Phase unwrapping is a key problem in many coherent imaging systems, such as synthetic aperture radar (SAR) interferometry. A general formulation for redundant integration of finite differences for phase unwrapping ~\citeravi{Costantini2010} was shown to produce a more reliable solution by exploiting redundant differential estimates. However, this technique requires a commercial linear programming solver for large-scale problems. For a linear cost function, we propose a method based on Dual Decomposition that breaks the given problem defined over a non-planar graph into tractable sub-problems over planar subgraphs. We also propose a decomposition technique that exploits the underlying graph structure for solving the sub-problems efficiently and guarantees asymptotic convergence to the globally optimal solution. The experimental results demonstrate that the proposed approach is comparable to the existing state-of-the-art methods in terms of the estimate with a better runtime and memory footprint.
\end{abstract}

\input{sec-intro}

\input{sec-relatedwork}

\input{sec-method}

\input{sec-results}

\input{sec-conclusion}

\input{sec-ack}

\bibliographystyle{plainnat}
\bibliography{references}

\input{supplementary}

\end{document}

%% file: sec-intro.tex
\section{Introduction} 

Phase unwrapping is the process of recovering unambiguous phase values from multi-dimensional phase data that are measured modulo $2\pi$ (wrapped data) and is affected by random noise and systematic disturbances. The difference between the measured and the actual phase has an ambiguity is integer multiples of $2\pi$.  For unwrapping purposes it is usually assumed that the sampling rate is adequate over most of the data set so that aliasing is avoided. In other words, the true absolute phase difference between two neighboring data points is generally less than $\pi$. This reduces the phase unwrapping problem into that of integration of the phase difference between neighboring data points under certain paths (in noisy conditions).  This problem is representative of a class of important imaging techniques such as interferometric synthetic aperture radar ~\citeravi{Allen2008}, optical interferometry ~\citeravi{Pandit:94}, magnetic resonance imaging ~\citeravi{4685486}, and diffraction tomography ~\citeravi{4157460}. The advent of synthetic aperture radar interferometry (InSAR) \citeravi{rosen2010} and numerous other applications spurred interest in developing reliable two-dimensional (2D) phase unwrapping algorithms. The most widely used technique for 2D phase unwrapping relies on network programming approaches, which formulates the problem as a Minimum Cost Flow (MCF) \citeravi{Costantini1998} and is applicable only for planar graphs.

Some of the factors that influence the quality of phase unwrapping for SAR interferometry include (1) signal interference due to atmospheric conditions; (2) temporal decorrelation due to the changes of the scattering characteristics at different times; and (3) geometric decorrelation due to different imaging geometries arising from the long distances between the repeated orbits of the same satellite. To address these factors, several techniques (\citenopar{Hooper2007};\citenopar{Costantini2010}; \citenopar{Shanker2010}) were developed to process multiple conventional interferograms collected at different times from the same area. A general formulation for redundant integration technique for multi-temporal phase unwrapping \citeravi{Costantini2010}  exploits redundant information obtained from any pairs of points (typically close together but not necessarily nearest neighbors), making it possible to obtain a solution  robust to outliers and noise both in 2D and  multidimensional cases. The general formulation includes standard phase unwrapping and finite difference integration techniques as special cases. 

Edgelist algorithm \citeravi{Shanker2010}  utilizes the information from the temporal dimension by replacing the MCF formulation with its dual  formulation, where  the  closed  loops  are  replaced  by  reliable  edges  as  the  base  construct.  This variant formulation enables the inclusion of  external geodesic measurements as constraints. A step-wise 3D algorithm \citeravi{Hooper2007} solves the multi-temporal unwrapping in multiple stages, first temporally and subsequently uses relaxed temporal solution to constrain the solution in spatial dimensions.

Similar to MCF formulation \citeravi{Costantini1998}, Edge list algorithm and Linear Programming (LP) based redundant arcs formulation \citeravi{Costantini2010}  exhibit total unimodularity which enables us to solve the integer programming problem efficiently. However, for large-scale interferograms, despite using commercial solvers, these formulations do not scale better than  MCF formulation (section \ref{results}). Also, the publicly available LP solvers  perform poorly when compared against commercial solvers in handling large datasets \citeravi{Cs2012}. In the rest of the paper, we refer to the LP based redundant arcs formulation  as LPRA.

Decomposition methods employ message-passing between sub-problems to iteratively improve the solution of large-scale optimization problems \citeravi{Santos2009}. However, their use is unexplored for phase unwrapping techniques to the best of our knowledge. Our contributions are: (1) We show how decomposition methods \citeravi{Rush2012} that break the given problem into tractable subproblems, can be employed to solve the LPRA efficiently (2) We propose a decomposition that asymptotically converges to a globally optimal solution and shows a significant runtime improvements over  LPRA. Further, the decomposed sub-problems only requires an MCF solver for which numerous efficient implementations are publicly available \citeravi{Kiraly2012}. In this paper, we focus on 2D interferograms with spatial redundant arcs. 

The remainder of the paper is organized as follows: The background and related work is discussed in section (\ref{background}).  An introduction to decomposition method based on Lagrangian relaxation and the necessary condition for asysmptotic optimality is described in sections (\ref{ourappraoch}) and (\ref{theory}) respectively. The validation methods and results are described in section (\ref{results}). Conclusion with future directions are discussed in section (\ref{futurework}).

%% file: sec-relatedwork.tex
\section{Background and Related work} \label{background}

The general phase unwrapping formulation is reviewed in section (\ref{PhunwForm}). A brief description of the conditions necessary for solving phase unwrapping as a MCF in section (\ref{netprog}). LPRA construction as a totally uni-modular formulation is described in section (\ref{redarcs}).

\subsection{Phase Unwrapping Formulation} \label{PhunwForm}

Consider a set of points in 2D and let the phase measured at a point $i$ be ${\psi_{i}}$. Phase Unwrapping involves extraction of the unwrapped phase value, $\Phi_i$, from ${\psi_{i}}$, which are related by 
{\small
\begin{gather}
\Phi_{i} = \Psi_{i} + 2\pi n_{i}
\end{gather}
}
where $n_{i}$ represent the integer number of the cycles that must be added to $\Psi_i$ to obtain the corresponding $\Phi_i$. In addition, it is convenient to define a directed graph $\mathcal{G} := (\mathcal{V}, \mathcal{E})$, whose nodes are the set of grid points $i$ and an edge $\edgem{i}{j} \in \mathcal{E}$ that connects the grid points $i$ and $j$. We define a new variable $\delta_{ij}$ to represent the integer flow along the directed edge \edge{i}{j} such that
\begin{gather}
n_{i} - n_{j} + \delta_{ij} - \delta_{ji} =  \delta^{\prime}_{ij} \text{ ,} ~ ~ \forall \edgem{i}{j} \in \mathcal{E} \\
\text{where,} \qquad \delta^{\prime}_{ij} = \Bigg[\frac{\psi_{i} - \psi_{j}}{2\pi}\Bigg]  \nonumber \\
 [.] ~~ \text{represents the nearest integer function} \nonumber
\end{gather}

\pagebreak The phase unwrapping can then be stated as,
\begin{gather}
\label{eq:objective}
\min_{\delta} \sum\limits_{\edgem{i}{j} \in \mathcal{E}} c_{ij}|\delta_{ij}|^m 
\end{gather}
\begin{gather}
\text{subject to: } \qquad \qquad \qquad \qquad \qquad \nonumber \\
n_{i} - n_{j} + \delta_{ij} - \delta_{ji} = \delta^{\prime}_{ij} \text{ ,} \qquad \forall ~ \edgem{i}{j} \in \mathcal{E} \label{eq:constraint}\\
\qquad \qquad  n_{i}, n_{j} \in Z^{+} \qquad \quad \forall  ~ i,j \in \mathcal{V} \nonumber\\
\qquad \qquad \delta_{ij}, \delta_{ji} \in \{0,1\} \qquad \text{ } \forall ~  \edgem{i}{j} \in \mathcal{E} \nonumber
\end{gather}
where the positive exponent {\it m} defines the selected metric, and $c_{ij}\in R$ representing the reliability of the preliminary estimates. When $m=1$ ($m=2$), the above problem reduces to a linear integer (quadratic) programming problem. The value of $n_{i}$ are determined only up to an additive constant because an additive constant leaves constraint (\ref{eq:constraint}) unaltered.

\subsection{Phase Unwrapping as Network Programming} \label{netprog}

When the graph $G$ considered in equation (\ref{eq:objective}) is planar, the constraint matrix is totally unimodular. In addition, the problem can be transformed to an equivalent MCF formulation in the dual network. Let $F$ be the faces in the planer embedding of $G$. Then, the constraints of the dual network is obtained by the set $F$ as follows. For each face $f \in F$, the edges constituting $f$ are traversed according to an uniformly chosen orientation and the constraints corresponding to the traversed edges are integrated together. As an example, consider the face shown below, comprised by the vertices $i, j, k$. Let the orientation be counter-clockwise. 
\begin{center}
{\small
\begin{tabular}{lll}
\multicolumn{3}{c}{\resizebox{1.95cm}{!}{\input{pictures/sample_delaunay.tikz}}}\\ 
\end{tabular}
}
\end{center}
The constraints of the edges $ij, jk, ki$ are
\begin{align}
	n_{i} - n_{j} + \delta_{ij} - \delta_{ji} = \delta^{\prime}_{ij}  \nonumber  \\
	n_{j} - n_{k} + \delta_{jk} - \delta_{kj} = \delta^{\prime}_{jk}	\label{left} \\
	n_{k} - n_{i} + \delta_{ki} - \delta_{ik} = \delta^{\prime}_{ki}  \nonumber 
\end{align}
Traversing the edges in the order $ij, jk, ki$ and integrating the constraints, we get the reformulated constraint for MCF as
{\small
\begin{equation}
\label{mcfconstraint}
\begin{split}
  (\delta_{ij} - \delta_{ji}) + (\delta_{jk} - \delta_{kj}) + (\delta_{ki} - \delta_{ik})\\
   =  \delta^{\prime}_{ij} + \delta^{\prime}_{jk} + \delta^{\prime}_{ki} 
\end{split}
\end{equation}
}
The integration path becomes
\begin{align}
n_{i} - n_{j} + \delta_{ij} - \delta_{ji} = \delta^{\prime}_{ij} \nonumber \\
n_{j} - n_{k} + \delta_{jk} - \delta_{kj} = \delta^{\prime}_{jk} 
\end{align}
Many different strategies \citeravi{Kiraly2012} exist to solve MCF formulation. Supplementary constraints could be used to reflect a priori considerations, or could be useful in initializing the algorithm for improving the runtime. 

\subsection{Redundant arcs construction as LP}
\label{redarcs}

MCF formulation is applicable only to planar graphs. In the case of a general graph $\mathcal{G}$, \citenopar{Costantini2010} showed that it is sufficient to choose the cycle space of $\mathcal{G}$ as the basis and to sum the constraints of each cycle (in an uniformly chosen orientation) to transform the constraint matrix to a totally unimodular matrix. A numerically efficient way to construct this constraint space is by using a spanning tree. If the spanning tree is $\mathcal{S}:=(\mathcal{V}, \mathcal{T})$ of $\mathcal{G}$, then each of the edges, $\edgem{i}{j} \in \mathcal{E}\backslash \mathcal{T}$ are called backedges. Each such backedge forms a unique cycle in $\mathcal{G}$ called fundamental cycle. 

This construction enables an optimal relaxation to the linear domain (for $m=1$) through the use of totally unimodular property. Computationally efficient algorithms exist \citeravi{opac-b1100407} to solve redundant LP method, although slower than MCF. Our contribution in this work uses this construction to develop an iterative technique that is computationally more effective.

%% file: pictures/sample_delaunay.tikz
\begin{tikzpicture}
\draw (0,0) node[anchor=north]{$i$}
  -- (0,2) node[anchor=south]{$k$}
  -- (2,2) node[anchor=south]{$j$}
  -- cycle;
  
\draw [->,line width=1pt] (0.45,1.45) arc[x radius=0.25cm, y radius =.25cm, start angle=-220, end angle=90];

\end{tikzpicture} 

%% file: sec-method.tex
\subsection{Our Approach} \label{ourappraoch}

This section is organized as follows. A brief review of Dual Decomposition technique is presented in section (\ref{dd}). In  section (\ref{theory}) and section (\ref{decomp}), our method with the properties of the sub-problems for runtime improvement and asymptotic optimality are discussed respectively. 

\subsection{Dual Decomposition} 
\label{dd}

We will first describe the general framework to illustrate methodology behind dual decomposition and then describe how this technique leads to an efficient way to solve phase unwrapping on a non-planar graph. Consider the following problem
\begin{gather}
\min\limits_{\bf x} f({\bf x})  \label{eq:refdd} \\
\text{subject to:} \quad {\bf x} \in Q \quad \nonumber
\end{gather}
where {\it Q} is a closed convex set, ${\bf x}$ is a vector and $f$ is any function of {\bf x}. Suppose, the objective function $f$ is linearly separable (as in equation (\ref{eq:objective})), that is $f({\bf x}) = \sum\limits_{i}f_i({\bf x}_{i})$. Then, the above problem can be transformed using auxiliary public variable $z$ as,
\begin{align}
\min\limits_{\{{\bf x}_{i}\}, z} \sum_{i} f_{i}({\bf x}_{i})  \label{eq:refddsim} \quad \quad \\
\text{subject to:} \quad  
{\bf x}_{i} \in Q_{i} \quad \forall i \nonumber \\
\quad {\bf x}_{i} = P_{i}{\bf z}  \quad \forall i \nonumber
\end{align}
where $Q_i$ is a closed convex set, such that $\cap_{i} Q_i = Q $. Vector ${\bf x}_i$ corresponding to each sub-problem (indexed by $i$) is the local copy of the public vector ${\bf z}$. $P_{i}$ is a binary matrix that maps the public vector $z$ into variables corresponding to each sub-problem ${\bf x}_i$. We can decouple the objective functions $f_i$ by relaxing the coupling constraint $x_{i} = P_{i}z$ using Lagrangian multipler ($\lambda$) to form the following dual function:
\begin{gather}
L({\bf z}, \lambda) = \min\limits_{\{{\bf x}_{i} \in Q_{i}\}, {\bf z}} \sum_{i} f_{i}({\bf x}_{i}) + \lambda^{T} \Big({\bf x}-P{\bf z}\Big)  \nonumber \\
\qquad \qquad \quad = \min\limits_{\{{\bf x}_{i} \in Q_{i}\}, {\bf z}} \sum_{i} \Big(f_{i}({\bf x}_{i}) + \lambda_{i}^{T}{\bf x}_{i}\Big) - \lambda^{T}P{\bf z}  \label{eq:lagrelax}
\end{gather}
Differentiating $L(z, \lambda)$ with respect to $z$ yields the condition $P^{T}\lambda = 0$. The dual function can now be solved independently for each $i$ given $\lambda$
\begin{gather}
L_{i}(\lambda_{i}) = \min\limits_{x_{i}} f_{i}({\bf x}_{i}) + \lambda_{i}^{T}{\bf x}_{i} \label{eq:subprob} \\ 
\text{subject to:} \quad {\bf x}_{i} \in Q_{i} \nonumber 
\end{gather}
The dual of the equation (\ref{eq:refdd}) then becomes 
\begin{gather}
\max_{\lambda} \quad L(\lambda) = \sum\limits_{i} L_{i}(\lambda_{i}) \label{eq:dual} \\
\text{subject to:} \quad P^{T}\lambda = 0 \nonumber
\end{gather}
The equation (\ref{eq:dual}) is convex and can be solved with the projected subgradient method \citeravi{Bertsekas2010}. We refer to equation (\ref{eq:dual}) as the dual decomposition master problem (DDMP). The projected subgradient algorithm is used to iteratively solve the DDMP algorithm, where at each iteration we first perform the decentralized optimization of the subproblems (in equation ~\ref{eq:subprob}) followed by a sub-gradient update of the lagrangian multipliers ($\lambda$) as follows: 
\begin{align}
\lambda_{i} \leftarrow \lambda_{i} + \alpha_{t} g_{\lambda_{i}}
\end{align}
Here, $g_{\lambda_{i}}$ is the subgradient of the objective function $L_i$ with respect to $\lambda_{i}$ and $\alpha_{t}$ denotes the step size along the direction of positive subgradient. In our case, the subgradient of $g_{\lambda}$ is simply $x_{i}^{*}$, which is the optimal point $x_{i}$ in equation (\ref{eq:subprob}). The lagrangians $\lambda$ are then projected onto the constraint space in equation~\ref{eq:dual} given by:
\begin{align}
\lambda_{i} \leftarrow \mathfrak{N}_{\lambda}( \lambda_{i} + \alpha_{t} x_{i}^{*}) \label{eq:projsubgrad}
\end{align}
where, $\mathfrak{N}_{\lambda}$ denotes the projection function onto the set $\{\lambda | P^{T} \lambda = 0\}$. This process is repeated until convergence of the DDMP objective.

\input{sec-theory}

\subsection{Phase Unwrapping via Dual Decomposition}
\label{decomp} In our approach, we exploit the {\it strong} duality properties of any linear program to develop an iterative algorithm for $\mathbb{PU}$. Assume that the constraints of any linear program can be divided into "easy" and "hard" constraints. More specifically in the context of phase unwrapping, the constraints that break the planarity property of the underlying graph is classified as "hard" constraint. We can then use $\mathbb{LR}$ to transfer some of the constraints into the objective function. Thus, the easy constraints define $A(\Gamma^{k})$. This leads us to the following theorem: 

\begin{theorem} {\bf (Cycle Decomposition):}
\label{th:main}
For any graph decomposition that covers $\mathcal{G}$, $\mathbb{LR} \leq \mathbb{PU}$. Furthur, the equality strictly holds when $\text{ }\cap_{k} \mathcal{A}(\Gamma^k) = \mathcal{A}(\Gamma)$.
\end{theorem}

We refer the reader to the Appendix~\ref{sup:proof} for a proof. Theorem~\ref{th:main} allows us to break the $\mathbb{LP}$ thereby enabling us to use the underlying planarity structure and still retain asympthotic optimality. The asymptotic nature of the algorithm stems from the fact that the lagrangian variables $\lambda$ in the $\mathbb{LR}$ problem is iteratively updated through sub-gradient ascent \citeravi{Rush2012}. 

In other words, to guarantee asymptotic optimality, we have to ensure that for each vertex pair belonging to an edge, there exists a path along a chosen spanning tree in atleast one of the subgraphs. We show an example decomposition of a non-planar subgraph $K_5$ in figure~\ref{tab:nonplanardecomp} in which the chosen spanning tree is shown in red. This ensures that the decomposition works on the $\mathbb{LP}$ formulation of the original problem \citeravi{Costantini2010}, thus inheriting asymptotic optimality. In addition, the planarity of each subgraph enables us to solve them efficiently as an MCF. 
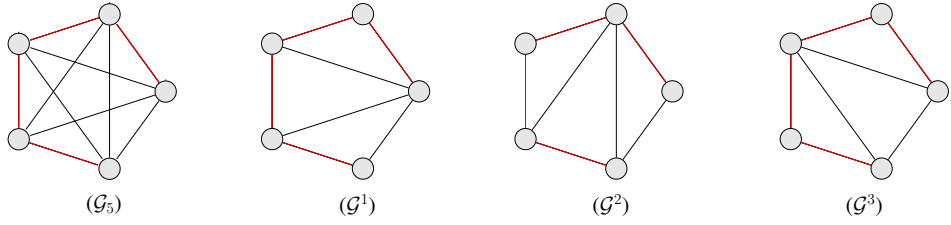
\begin{figure}
\begin{center}
{\small
\newcolumntype{C}{ >{\centering\arraybackslash} m{8cm} }
\scalebox{0.4}{
\begin{tabular}{CCCC}
\input{pictures/k5.tikz} \qquad \qquad \qquad & \input{pictures/k5_1.tikz} \qquad \qquad \qquad & \input{pictures/k5_2.tikz} \qquad \qquad \qquad & \input{pictures/k5_3.tikz} \\ \\
{\huge ($\mathcal{G}_5$)} & {\huge ($\mathcal{G}^1$)} & {\huge ($\mathcal{G}^2$)} & {\huge ($\mathcal{G}^3$)} \\ \\
\label{subgraph}
\end{tabular}}
\caption{\label{tab:nonplanardecomp} \small $\mathcal{G}_5$ is an example of phase unwrapping on a non-planar graph. The nodes and the edges of the graph correspond to the measured wrapped phase and the constraints respectively. The sub-graphs $\mathcal{G}^1$, $\mathcal{G}^2$ and $\mathcal{G}^3$ show an example decomposition of the non-planar graph. Red colored edges in $\mathcal{G}_5$ and each of subgraph $\mathcal{G}^{1}$, $\mathcal{G}^{2}$ and $\mathcal{G}^{3}$ corresponds to the chosen spanning tree. }
}
\end{center}
\end{figure}

\begin{center}
\noindent
\begin{varwidth}{\dimexpr\linewidth-40\fboxsep-10\fboxrule\relax}
\begin{algorithm}[H]
\caption[l]{\tabular[t]{@{}l@{}} Phase Unwrapping with Redundant Arcs via \\ Dual Decomposition \endtabular }\label{ddphw}
\DontPrintSemicolon
\textbf{Initialize}: dual variables $\lambda$ that satisfies $\text{        }\sum_{k  \in G_{e}(i,j)} \lambda^{k}_{ij} = 0 \text{ } \forall \text{ } ij \in \mathcal{E}$  (e.g., $\lambda = 0$)\\
\While{Stopping criteria not met}{
 \Comment*[l]{\small Optimize each subproblem separately using MCF}
 \For{$k=1$ to $\mathcal{K}$}{
        $ \delta^{k\star} =\argminF_{\delta^{k} \in \mathcal{A}(\Gamma^k)} \sum_{k} (c^{k} + \lambda^{k})^{T}\delta^{k}$
 }
 \          \\
 \Comment*[l]{\small Update the dual variables}
 \For{$k=1$ to $\mathcal{K}$}{
		\For{$(i,j) \in \mathcal{G}_e(i,j)$}{
        	$\lambda_{ij}^{k} += \alpha_{t} \Bigg(\delta_{ij}^{k\star} - \mathlarger{\frac{1}{|\mathcal{G}_e(i,j)|}}\mathlarger{\sum\limits_{q\in\mathcal{G}_e(i,j)}}\delta_{ij}^{q*} \Bigg)$
        }
 }
}
\end{algorithm}
\end{varwidth}%
\end{center}
\vspace{0.02in}

The constraint on the cost function (in definition ~\ref{dfn:decomp}) can be satisfied by distributing the cost ($c_{ij}$) uniformly among subgraphs. The algorithm (\ref{ddphw}) describes the pseudo-code for Dual Decomposition based phase unwrapping. It can be easily shown that the algorithm (\ref{ddphw}) corresponds to a projected subgradient ascent algorithm for solving $\mathbb{LR}$ problem.

The algorithm (\ref{ddphw}) consists of two main steps:
\begin{itemize}
\item Line 4: In the first step, the lagrangian variables are fixed to optimize each subproblem ($f^k(\lambda^k$) in definition~\ref{dfn:decomp}) as a MCF independently.

\item Line 8-9: In the second step, the dual variables are updated using subgradient ascent,  such that the consistency constraint (as in definition~\ref{dfn:decomp}) for local flow variables belonging to each subproblem is satisfied.
\end{itemize}

These steps are repeated until convergence. We use the relative change in the dual objective value as the convergence criteria for the algorithm (\ref{ddphw}), since it is expected to plateau near the dual optimum.

%% file: sec-theory.tex
\subsection{Definitions} \label{theory}

In this section, we define the transformed phase unwrapping following the LPRA construction and the Lagrangian relaxation for the phase unwrapping problem more formally. 

Without loss of generality, assume that the graph $\mathcal{G}$ is connected. Let $C$ be a sequence of arcs forming a closed path, and $\Gamma$ be a set of independent closed paths that span the whole cycle-space. Then, LPRA \citeravi{Costantini2010} provides a numerically efficient way to solve the original phase unwrapping problem, by summing for each $C\in \Gamma$ the equation~\ref{eq:constraint} that corresponds to arcs $(i,j)$ belonging to $C$ that transforms the constraints space to the following set of equations, denoted $\mathcal{A}(\Gamma)$:

{
\small
\begin{gather}
\mathcal{A}(\Gamma):\qquad \sum_{(i, j) \in C} (\delta_{ij} - \delta_{ji}) = \sum_{(i, j) \in C} \delta^{\prime}_{ij}, \qquad C\in \Gamma; \qquad \forall \delta_{ij}, \delta_{ji} \in \{0,1\} \label{eq:lincnstr} 
\end{gather}
}

\begin{defn}
\label{def:TU}
The transformed phase unwrapping on $\mathcal{G}(\mathcal{V}, \mathcal{E})$ with cycle-space $\Gamma$, denoted $\mathbb{PU}$, is: 
\begin{gather*}
\mathbb{PU}:=\min\limits_{\substack{\delta \in \mathcal{A}(\Gamma) \\ \delta \in \{0,1\}}} \quad \sum_{ij\in\mathcal{E}} c_{ij}\delta_{ij}
\end{gather*}

The linear minimization problem, denoted $\mathbb{LP}$, is obtained by linear relaxation of the integer flows:

\begin{gather*}
\mathbb{LP}:=\min\limits_{\substack{\delta \in \mathcal{A}(\Gamma) \\ {\color{blue} \delta \in [0,1]}}} \quad \sum_{ij\in\mathcal{E}} c_{ij}\delta_{ij}
\end{gather*}
\end{defn}

\begin{defn}
\label{dfn:decomp}
With a little abuse of exponent notation, let $\mathcal{H} = \{\mathcal{G}^{k}(\mathcal{V}^{k}, \mathcal{E}^{k}): k \in [1, \mathcal{K}]\}$ be a decomposition of a graph $\mathcal{G}$, where $\mathcal{K}$ is the cardinality of the set $\mathcal{H}$ that covers the graph $\mathcal{G}$ such that $\cup \mathcal{E}^{k} = \mathcal{E}$ and $\mathcal{V}^{k} = \mathcal{V}$. Let $\mathcal{G}_{e}(i,j)$ be the set of all subgraphs that contain edge $(i, j)$. Let $\mathcal{A}(\Gamma^k)$ be the constraint space defined on each subgraph $\mathcal{G}^k$, with cycle space $\Gamma^{k} \subset \Gamma $ (as in equation~\ref{eq:lincnstr}).

Assume that we introduce local flow variables and the corresponding cost function $\{\delta^k_{ij}, c_{ij}^{k}: ij \in \mathcal{E}_{k}\}$ for each subproblem $k$, with additional consistency constraints $\sum_{k} c_{ij}^{k} = c_{ij}$ and $\delta^{k}_{ij}=\delta_{ij}$. We can then define the Lagrangian relaxation for the $\mathbb{LP}$ by introducing a vector of lagrangian variable $\lambda$ for every consistency constraint on $\delta_{ij}^k$, denoted $\mathbb{LR}$, as:

\begin{gather*}
\mathbb{LR}:= \max_{\lambda} \sum_{k} f^{k}(\lambda^k)\\
f^{k}(\lambda^k) = 
\min\limits_{\substack{\delta^{k} \in \mathcal{A}(\Gamma^k) \\ {\color{blue} \delta^{k}, \delta \in [0,1]}}} \sum_{ij \in \mathcal{E}^k} c^{k}_{ij}{\delta}^{k}_{ij} + \lambda_{ij}^{k} \Big({\delta}^{k}_{ij} - {\delta}_{ij} \Big)
\textbf{\textbf{\textbf{}}}\end{gather*}

\end{defn}

We can now state the theorem~\ref{th:tum} behind LPRA algorithm's construction, which appears as Theorem 3.4c in ~\citeravi{Kavitha2009}.

\begin{theorem}
\label{th:tum}
{\bf (Tight Relaxation):} Constraint matrix $\mathcal{A}$ defined on the cycle space $\Gamma$ is totally unimodular. Thus, if either $\mathbb{PU}$ or $\mathbb{LP}$ has a finite optimal value, then so does the other, and their optimal values coincide. 
\end{theorem}

Theorem~\ref{th:tum} enables us to solve the original phase unwrapping on a non-planar graph as a linear programming problem. In other words, this theorem states that the linear relaxation $\mathbb{LP}$ is {\it tight}. We now develop an iterative algorithm for phase unwrapping.

%% file: pictures/k5.tikz
\begin{tikzpicture}[node distance = 3cm and 4cm,
    				el/.style = {inner sep=2pt, align=left},
					every label/.append style = {font=\tiny},
					every  edge/.append style = {draw, -stealth', shorten > = 1pt,
                    font=\footnotesize, inner sep=2pt, auto, align=left, sloped}, scale=0.5]
  \foreach \x in {1,...,5}{%
    \pgfmathparse{(\x-1)*360/5}
    \node[draw,circle,inner sep=0.25cm, fill=gray!20!] (N-\x) at (\pgfmathresult:5.4cm) [thick] {};
  }
  \foreach \x in {1,...,5}{%
    \foreach \y in {\x,...,5}{%
        \path (N-\x) edge[thick,-] (N-\y);
  }
  }

  \path (N-1) edge[thick, -] (N-2);
  \path (N-2) edge[thick, -] (N-3);
  \path (N-3) edge[thick, -] (N-4);
  \path (N-4) edge[thick, -] (N-5);

  \path (N-1) edge[red, thick, -] (N-2);
  \path (N-2) edge[red, thick, -] (N-3);
  \path (N-3) edge[red, thick, -] (N-4);
  \path (N-4) edge[red, thick, -] (N-5);
\end{tikzpicture}

%% file: pictures/k5_1.tikz
\begin{tikzpicture}[ultra thick, scale=0.5]
  \foreach \x in {1,...,5}{%
    \pgfmathparse{(\x-1)*360/5}
    \node[draw,circle,inner sep=0.25cm, fill=gray!20!] (N-\x) at (\pgfmathresult:5.4cm) [thick] {};
  }

  \path (N-1) edge[thick, -] (N-2);
  \path (N-1) edge[thick, -] (N-3);
  \path (N-1) edge[thick, -] (N-4);
  \path (N-1) edge[thick, -] (N-5);
  \path (N-2) edge[thick, -] (N-3);
  \path (N-3) edge[thick, -] (N-4);
  \path (N-4) edge[thick, -] (N-5); 

  \path (N-1) edge[thick, -] (N-2);
  \path (N-2) edge[thick, -] (N-3);
  \path (N-3) edge[thick, -] (N-4);
  \path (N-4) edge[thick, -] (N-5);

  \path (N-1) edge[red, thick, -] (N-2);
  \path (N-2) edge[red, thick, -] (N-3);
  \path (N-3) edge[red, thick, -] (N-4);
  \path (N-4) edge[red, thick, -] (N-5); 
\end{tikzpicture}

%% file: pictures/k5_2.tikz
\begin{tikzpicture}[ultra thick, scale=0.5]
  \foreach \x in {1,...,5}{%
    \pgfmathparse{(\x-1)*360/5}
    \node[draw,circle,inner sep=0.25cm, fill=gray!20!] (N-\x) at (\pgfmathresult:5.4cm) [thick] {};
  }

  \path (N-2) edge[thick, -] (N-3);
  \path (N-2) edge[thick, -] (N-4);
  \path (N-2) edge[thick, -] (N-5);
  \path (N-2) edge[thick, -] (N-1);
  \path (N-3) edge[thick, -] (N-4);
  \path (N-4) edge[thick, -] (N-5);
  \path (N-5) edge[thick, -] (N-1);  

  \path (N-1) edge[thick, -] (N-2);
  \path (N-2) edge[thick, -] (N-3);
  \path (N-3) edge[thick, -] (N-4);
  \path (N-4) edge[thick, -] (N-5);

  \path (N-1) edge[red, thick, -] (N-2);
  \path (N-2) edge[red, thick, -] (N-3);
  \path (N-3) edge[red, thick, -] (N-4);
  \path (N-4) edge[red, thick, -] (N-5);

\end{tikzpicture}

%% file: pictures/k5_3.tikz
\begin{tikzpicture}[ultra thick, scale=0.5]
  \foreach \x in {1,...,5}{%
    \pgfmathparse{(\x-1)*360/5}
    \node[draw,circle,inner sep=0.25cm, fill=gray!20!] (N-\x) at (\pgfmathresult:5.4cm) [thick] {};
  }

  \path (N-3) edge[thick, -] (N-4);
  \path (N-3) edge[thick, -] (N-5);
  \path (N-3) edge[thick, -] (N-1);
  \path (N-3) edge[thick, -] (N-2);
  \path (N-4) edge[thick, -] (N-5);
  \path (N-5) edge[thick, -] (N-1);
  \path (N-1) edge[thick, -] (N-2);  

  \path (N-1) edge[thick, -] (N-2);
  \path (N-2) edge[thick, -] (N-3);
  \path (N-3) edge[thick, -] (N-4);
  \path (N-4) edge[thick, -] (N-5);

  \path (N-1) edge[red, thick, -] (N-2);
  \path (N-2) edge[red, thick, -] (N-3);
  \path (N-3) edge[red, thick, -] (N-4);
  \path (N-4) edge[red, thick, -] (N-5);
\end{tikzpicture}

%% file: sec-results.tex
\section{Results and Discussions} \label{results}

{
\captionsetup[subfloat]{captionskip=-0.2cm}
\captionsetup[subfigure]{labelformat=empty}
\begin{figure} [!ht]%
	\hspace{-1cm}
    \vspace{-0.6cm}
    \subfloat[Phase Image A]{\includegraphics[scale=0.3]{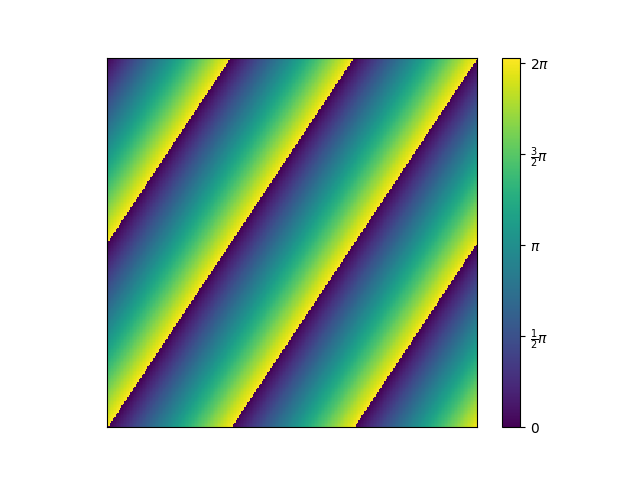}} %
    \subfloat{{\includegraphics[scale=0.3]{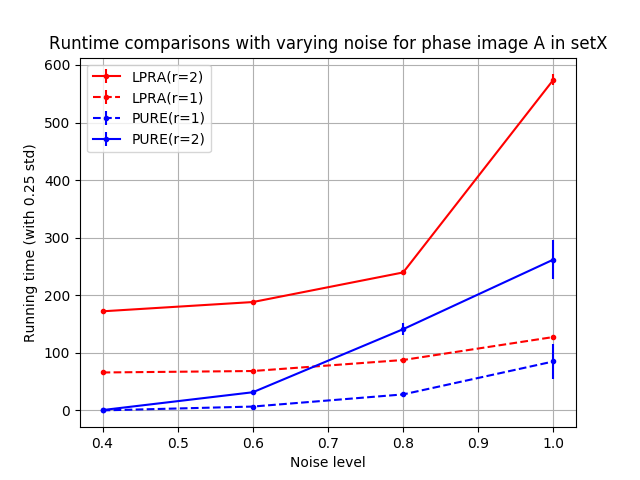} }}%
    \subfloat{{\includegraphics[scale=0.3]{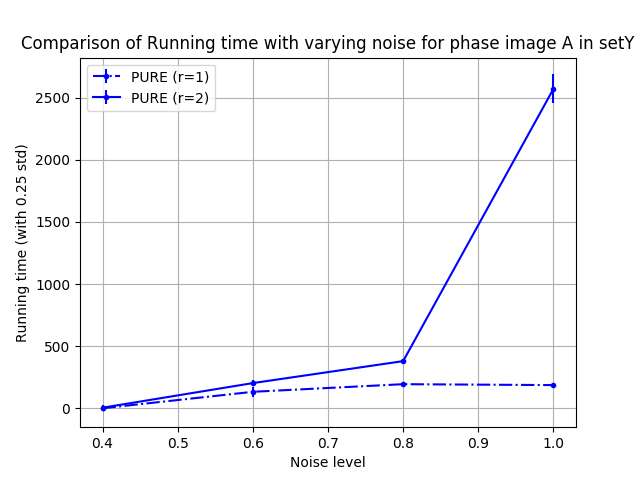} }} \\
    
	\hspace{-1cm}
    \vspace{-0.6cm}
    \subfloat[Phase Image B]{{\includegraphics[scale=0.3]{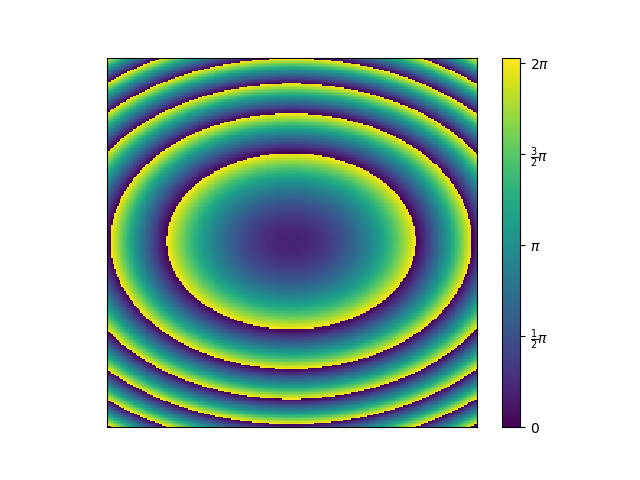} }}%
    \subfloat{{\includegraphics[scale=0.3]{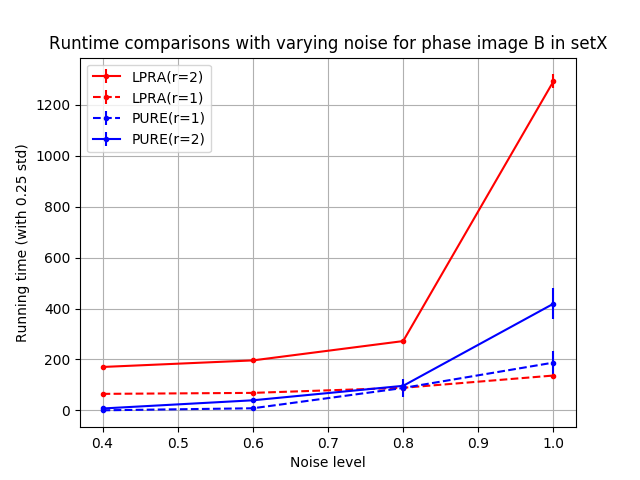} }}%
    \subfloat{{\includegraphics[scale=0.3]{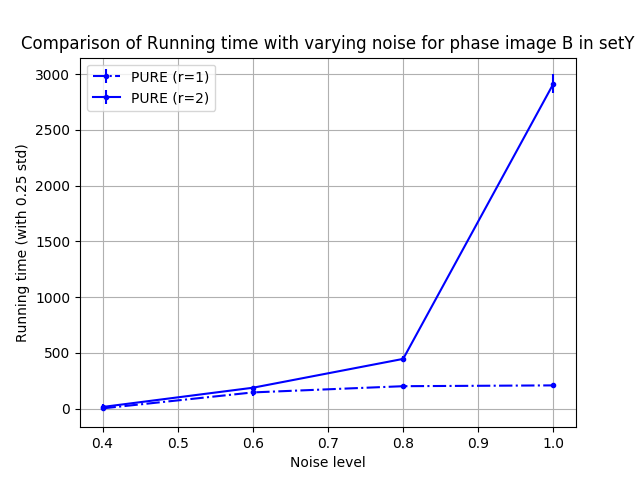} }} \\
    
	\hspace{-1cm}
    \vspace{-0.6cm}
    \subfloat[Phase Image C]{{\includegraphics[scale=0.3]{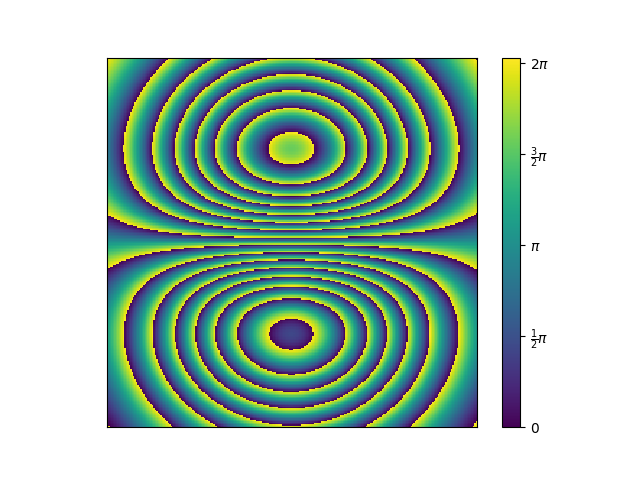} }}%
    \subfloat{{\includegraphics[scale=0.3]{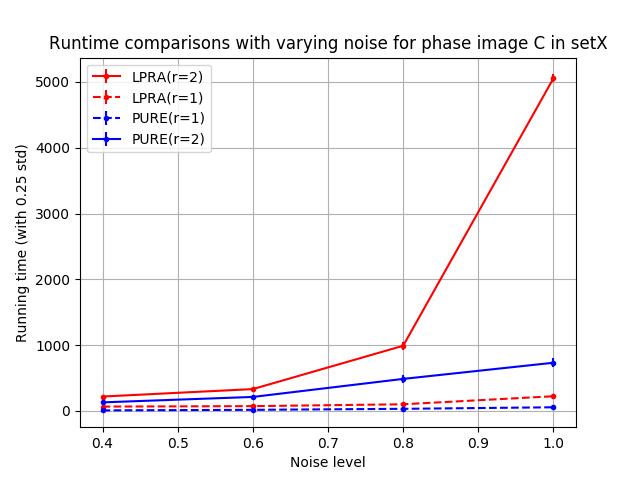} }}%
    \subfloat{{\includegraphics[scale=0.3]{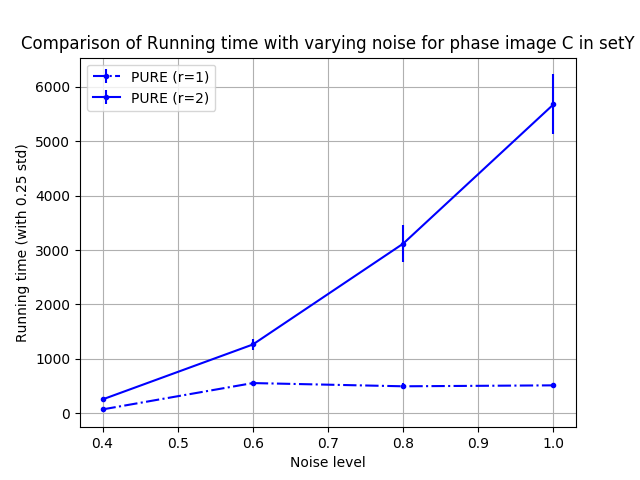} }} \\
    
	\hspace{-1cm}
    \vspace{-0.6cm}
    \subfloat[Phase Image D]{{\includegraphics[scale=0.3]{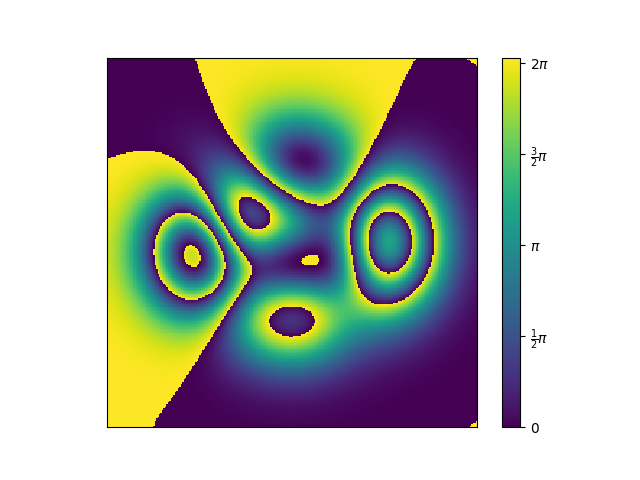} }}%
    \subfloat{{\includegraphics[scale=0.3]{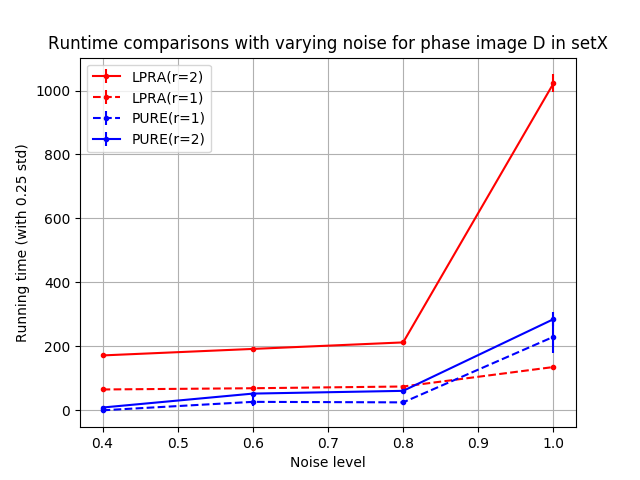} }}%
    \subfloat{{\includegraphics[scale=0.3]{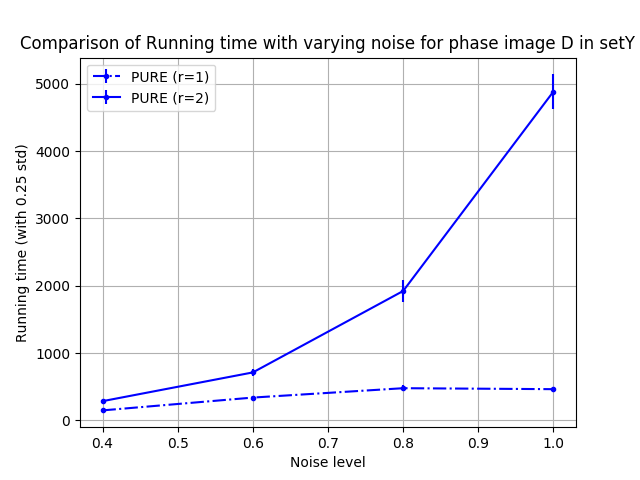} }} 
    \vspace{6mm}
    \caption{The first column above shows the wrapped phase images (A,B,C \& D) without noise used in all our experiments. The plots in the second and third column show the running time (in seconds) comparison between PURE and LPRA for variable noise levels in the input. The value r corresponds to the redundant arcs \{1, 2\}. In all of the plots above, the dotted and thick line correspond to running time with redundant arcs 1 \& 2 respectively.}%
   
    \label{fig:unwrapped_sim}%
    \vspace{-4mm}
\end{figure} 
}
\squeezeup
{\bf Methods Compared:} We first evaluate the performance of our method (henceforth referred to as PURE) using two different MCF algorithms, namely cost-scaling ~\citeravi{Goldberg1995} and the network simplex algorithm ~\citeravi{Orlin1997}. The cost-scaling algorithm is a primal-dual method that applies a successive approximation scheme by scaling the cost. The network simplex algorithm is a specialized version of the simplex algorithm for linear programming that exploits the network structure of the MCF problem and performs the basic operations directly on the graph representation. RelaxIV ~\citeravi{Bertsekas1994}, a dual ascent based algorithm is another notable solving technique for MCF. However RelaxIV is not directly applicable to PURE, since it assumes an integral cost.

Our primary empirical results is the comparison of PURE against LPRA with variable redundant arcs. In order for the evaluation to be implementation independent, we report the running time for each algorithm as the time taken by the respective solvers. For completeness, we also present the quality improvements we obtain by using the redundant arcs compared to LPRA (with redundant arcs distance = 2).

{\bf Experimental Setup:} We used a commercially solver Gurobi (Version 6.5.1) to solve the LPRA formulation. For solving MCF, we used the Cost-scaling code of Goldberg and Cherkassky, an efficient authoritative implementation of the algorithm. In addition, we used the Network scaling algorithm available in the MCFClass project ~\citeravi{Frangioni2006}, a common C++ interface for several MCF solvers implementation.

The cost functions (in definition~\ref{dfn:decomp}) is used to direct the placement of phase cycle jumps using quality measures derived from the data itself. The cost function is a crucial parameter that directly influences the quality of the solution. For our experiments, we used the smooth cost function as defined in SNAPHU \citeravi{Chen2001}. In all our experiments on PURE, we scaled the cost in the objective function between [0, 1] and used a constant learning rate ($\alpha_t$ in~\ref{ddphw}) until change in objective value is below a threshold of $0.02$ and a decaying learning rate ($\alpha\leftarrow\frac{\alpha}{2}$) from there on with a convergence threshold of $0.001$. 

\captionsetup[subfloat]{captionskip=0.001cm}
\captionsetup[subfigure]{labelformat=empty}
\begin{figure} [!ht]
    \subfloat[(a)]{\label{fig:solver_comp} \includegraphics[scale=0.45]{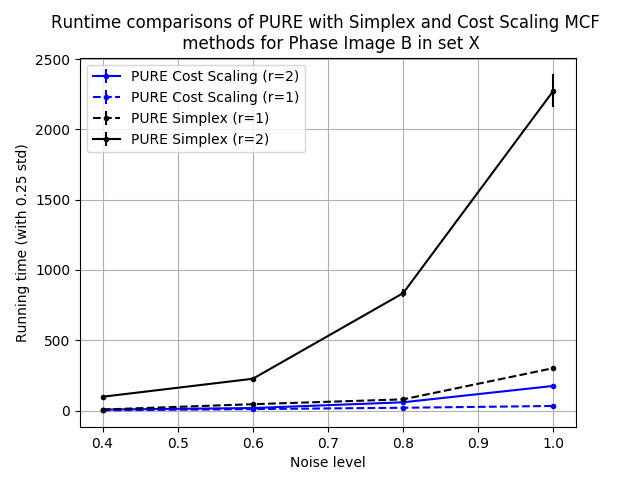}}%
    \subfloat[(b)]{\label{fig:incons}  \includegraphics[scale=0.45]{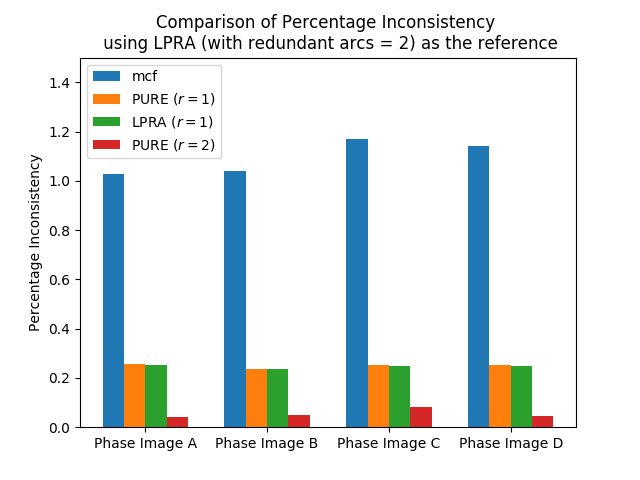}}%
    \caption{Figure (a) shows the Running time (in seconds) comparison of PURE with cost-scaling and simplex MCF algorithm. Figure (b) shows comparison of percentage inconsistency with unit noise variance using LPRA (with redundant arcs 2) as reference on setX. Value r in the legend corresponds to the redundant arcs length.}%
\end{figure} 

{\bf Datasets:} The experiment is conducted on a set of four SAR interferograms. For each of the interferogram, we generated $40$ different instance, $10$ instances each for every noise level from $0.4$ to $1$ with an increment of $0.2$. This allows us to study the running time performance of the algorithm with varying noise levels. We conduct the same set of experiments on two different resolutions $256\times256$ \& $768\times768$ (henceforth referred to as setX, setY respectively) to study the effect of constraint set size on the running time. The number of constraints which depend on redundant arcs ranges from ~$100k$ to over $500k$ for setX and from ~$1$ million to over $5$ million for setY, which can be quite challenging to solve even for a commercial solver as we show later.

\begingroup
\setlength{\intextsep}{0pt}
\setlength{\columnsep}{15pt}
{\bf Main Results:} In the first set of experiments, we conduct an empirical evaluation of the running time for PURE using cost-scaling and simplex algorithms. Figure ~\ref{fig:solver_comp} depicts the running time comparison as we vary the noise level for both redundant arcs \{1,2\} on setX. This plot is representative of the typical performance obtained in the setY instances. We observed that the performance of cost-scaling algorithm was more time efficient when compared to simplex algorithm for all our instances on setX. We are not aware of a theoretical interpretation of the performance difference at the time of writing the paper. For all our future experiments, we use PURE with the cost-scaling algorithm for solving the MCF subproblems.

In the second set of experiments, we validate the scalability of PURE by comparing the running time with LPRA. For this experiment, we predefined the planar decomposition for each redundant arcs level \{1, 2\} that satisfies the condition for asymptotic optimality as in Theorem~\ref{th:main}. First, we note that although all the problem instances in setX were solved by both the algorithms, PURE was several magnitudes faster in running time compared to LPRA. Secondly, LPRA failed to solve any of the instances in setY due to the large constraint set. To the best of our knowledge, PURE is the first algorithm to have solved constraint set of this scale for phase unwrapping with redundant arcs. 

\begin{wrapfigure}{r}{0.4\textwidth}
  \centering
    \includegraphics[width=0.35\textwidth]{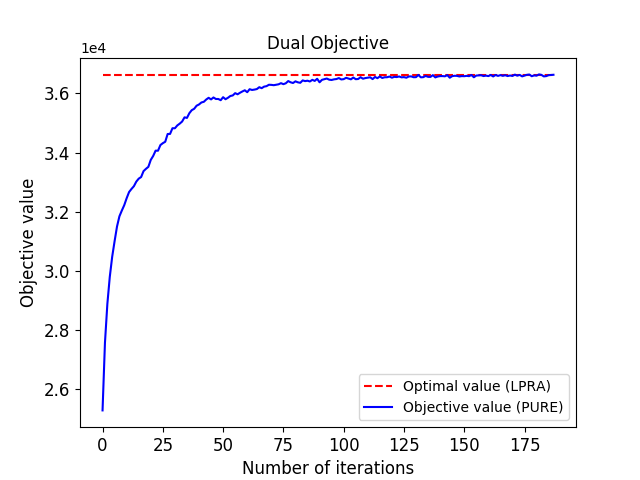}
  \caption{{\small Evolution of dual objective value for an instance of wrapped phase B in setX.}}
  \label{fig:learning_rate}
\end{wrapfigure}

In order to further quantitatively evaluate the performances of different phase unwrapping methods, we also present the comparison of inconsistency measure in Figure~\ref{fig:incons} with LPRA with redundant arcs distance = 2 as reference. The inconsistency measure is defined as the percentage number of reconstructed phase cycle that matches with ground truth phase cycle. As expected, LPRA and PURE show significant improvement over the MCF based approach. In addition, PURE is comparable in performance to that of LPRA with significantly improved running time performance. 

\endgroup

%% file: sec-conclusion.tex
\section{Conclusion and Future Work} \label{futurework}

We proposed a scalable algorithm for phase unwrapping with redundant arcs and provide empirical evaluation of our performance on simulate SAR interferograms. In future, we aim to extend this work to the temporal dimension of the interferograms. The rate of convergence and its analysis \citeravi{Rush2012} depend on the update schedule used by the algorithm. We would also like to analyze the rate of convergence under different variants of subgradient ascent techniques.

%% file: sec-ack.tex
\subsection*{Acknowledgments}
We would like to thank Piyush Agram for introducing us to Phase Unwrapping and Redundant Arcs technique and for the numerous discussion that shaped the problem definition. In addition, the authors would also like to thank BalaSanjeevi and Arun Viswanathan for reviewing the manuscript. This  work  was  carried  out  at  the  Jet  Propulsion  Laboratory,  California  Institute  of  Technology,  under  a  contract with  the  National  Aeronautics  and  Space  Administration. All rights reserved.

%% file: supplementary.tex
\section*{Appendix}
\appendix
\section{Proofs for Theoretical Results}
\label{sup:proof}

\begin{theorem*} {\bf (Cycle Decomposition):}
For any graph decomposition that covers $\mathcal{G}$, $\mathbb{LR} \leq \mathbb{PU}$. Furthur, the equality strictly holds when $\text{ }\cap_{k} \mathcal{A}(\Gamma^k) = \mathcal{A}(\Gamma)$.
\end{theorem*}
\begin{proof}
Let us define $\mathbb{\overline{LP}}$ to facilitate the proof. From definition~\ref{dfn:decomp}, $\Gamma^{k} \subset \Gamma$ which implies $\mathcal{A}(\Gamma) \subseteq \cap_{k}\mathcal{A}(\Gamma^{k})$. Let $B(\Gamma)$ denote the additional set of constraints in $\mathcal{A}(\Gamma)$ such that, $\mathcal{A}(\Gamma) = \cap_{k}\mathcal{A}(\Gamma^{k}) \cap B(\Gamma)$. Then $\mathbb{LP}$ can be transformed to the following form: 

\begin{gather*}
\mathbb{LP}:=\min\limits_{\substack{\forall \text{k, }  \delta^{k} \in \mathcal{A}(\Gamma^{k});\text{ }\delta^{k} = \delta \\ {\color{blue} \delta \in \mathcal{B}(\Gamma)};\text{ }\delta, \delta^{k} \in [0,1]}} \quad \sum_{k}  c^{kT} \delta^{k} \\
\text{where, }\sum_{k  \in G_{e}(i,j)} c_{ij}^{k} = c_{ij} \qquad  \forall ij \in \mathcal{E}
\end{gather*}

We can then define $\mathbb{\overline{LP}}$ as follows:
\begin{gather*}
\mathbb{\overline{LP}}:=\min\limits_{\substack{\forall \text{k, } \delta^{k} \in \mathcal{A}(\Gamma^{k});\text{ }\delta^{k} = \delta \\ \text{ }\delta, \delta^{k}  \in [0,1]}} \quad \sum_{k} c^{kT} \delta^{k} 
\end{gather*}

We note that $\mathbb{\overline{LP}}\leq \mathbb{LP}$ follows from the fact that the constraint space of $\mathbb{LP}$ is contained in the constraint space of that of $\mathbb{\overline{LP}}$. Equality constraint strictly holds when $\mathcal{B}(\Gamma)=R^{n}$, which implies $\mathcal{A}(\Gamma) = \cap_{k}\mathcal{A}(\Gamma^{k})$. For ease of notation, we use $A^{k}\delta^{k}-b^{k}\leq 0$ to denote $\{\delta^{k}:  \delta^{k} \in \mathcal{A}(\Gamma^{k}), \delta^{k} \in [0,1]\}$. Let us know consider the dual of $\mathbb{\overline{LP}}$, denoted dual-$\mathbb{\overline{LP}}$:

\begin{gather*}
\text{dual-}\mathbb{\overline{LP}} := \max_{\mu, \lambda}\text{ }\min_{\delta^k, \delta} \sum_{k} c^{kT} \delta^{k} + \mu^{kT}(A^k \delta^{k} - b^{k}) + \lambda^{kT}(\delta^{k} - \delta)
\end{gather*}

Upon further simplifying, we have
\begin{gather*}
\text{dual-}\mathbb{\overline{LP}} := \max_{\mu \leq 0, \lambda} \sum_{k} -\mu^{kT} b^{k} \\
\text{subject to:} \qquad c^{k} + A^{kT}\mu^{k} + \lambda^{k} = 0 \quad \forall\text{ }k \\
\qquad \qquad \qquad \qquad \qquad \text{  } \sum_{k  \in G_{e}(i,j)} \lambda^{k}_{ij} = 0 \quad \forall \text{ } ij\in\mathcal{E}
\end{gather*}

Using linear programs strong duality property, $\mathbb{\overline{LP}}=\text{dual-}\mathbb{\overline{LP}}$. Let $(\mu^{*}, \lambda^{*})$ be the optimal solution of $\text{dual-}\mathbb{\overline{LP}}$. Then evaluating $\mathbb{LR}$ at $\lambda^{*}$, denoted $\mathbb{LR}_{\lambda^{*}}$, we have

\begin{align*}
\mathbb{LR}_{\lambda^{*}} &= \sum_{k} \min_{A^{k} \delta^{k} \leq d^{k}} [c^{kT}\delta^{k} + \lambda^{*kT} \delta^{k}] \\
&= \sum_{k} \min_{A^{k} \delta^{k} \leq d^{k}} -\mu^{*kT} A^{k} \delta^{k} \geq \sum_{k} -\mu^{*kT} b^{k}
\end{align*}

Thus, $\mathbb{LR}_{\lambda^{*}} \geq \text{dual-}\mathbb{LP} = \mathbb{LP}$. We can also easily show that $\mathbb{LR} \leq \mathbb{LP}$. Thus $\mathbb{LR} = \mathbb{LP}$, when $\mathcal{A}(\Gamma) = \cap_{k}\mathcal{A}(\Gamma^{k})$ and it follows from theorem~\ref{th:tum}, $\mathbb{LR} = \mathbb{PU}$.

A variant of this proof was presented in ~\cite{Santos2009}, however to the best of our knowledege this is the first extension to TUM.

\end{proof}